%% file: main.tex
\documentclass[conference]{IEEEtran}

\IEEEoverridecommandlockouts
\usepackage{graphicx}
\usepackage{amsmath,amssymb,amsthm}
\usepackage{textcomp}
\usepackage{cite}
\usepackage{mathtools}
\usepackage{subfig}
\usepackage[nolist]{acronym}
\usepackage[hidelinks]{hyperref}
\usepackage{cleveref}
\usepackage{relsize}
\usepackage{float}
\usepackage[ruled,vlined]{algorithm2e}
\usepackage{tabularx}
\usepackage{booktabs}
\usepackage{mathtools}
\usepackage{pgf,tikz}
\usepackage{pgfplots}
\usetikzlibrary{arrows,positioning} 
\usetikzlibrary{decorations.text}
\usepgfplotslibrary{external}
\tikzsetexternalprefix{external_figs/}
\usepackage{contour}

\graphicspath{{./graphics/}}
\DeclareMathOperator*{\argmin}{arg\,min}
\theoremstyle{definition}

\theoremstyle{definition}

\theoremstyle{definition}

\newtheorem{lemma}{Lemma}
\newcommand{\Sc}{S_{\mathcal{C}_i}}
\newcommand{\Scb}{S_{\overline{\mathcal{C}}_i}}

\newcommand*{\la}{\big \lvert}%
\newcommand*{\ra}{\big \rvert}%
\DeclarePairedDelimiter{\ceil}{\lceil}{\rceil}
\definecolor{blue33D}{RGB}{0, 0, 175}
\newenvironment{reply}{}{}
\input{tikzsetup}

\newif\ifshowtikz
\showtikztrue
\let\oldtikzpicture\tikzpicture
\let\oldendtikzpicture\endtikzpicture

\renewenvironment{tikzpicture}{%
	\ifshowtikz\expandafter\oldtikzpicture%
	\else\comment%   
	\fi
}{%
	\ifshowtikz\oldendtikzpicture%
	\else\endcomment%
	\fi
}

\begin{document}

\author{Bashar Tahir, Stefan Schwarz, and Markus Rupp \\
	
	\thanks{Bashar Tahir and Stefan Schwarz are with the Christian Doppler Laboratory for Dependable Wireless Connectivity for the Society in Motion. The financial support by the Austrian Federal Ministry for Digital and Economic Affairs and the National Foundation for Research, Technology and Development is gratefully acknowledged.}

	Institute of Telecommunications, Technische Universit\"{a}t Wien, Vienna, Austria \\
%	Email: \{bashar.tahir, stefan.schwarz, markus.rupp\}@tuwien.ac.at
}

\title{Outage Analysis of Uplink IRS-Assisted \\ NOMA under Elements Splitting}

\maketitle
\begin{abstract}
In this paper, we investigate the outage performance of an \ac{IRS}-assisted \ac{NOMA} uplink, in which a group of the surface reflecting elements are configured to boost the signal of one of the \acp{UE}, while the remaining elements are used to boost the other \ac{UE}. By approximating the received powers as Gamma random variables, tractable expressions for the outage probability under \ac{NOMA} interference cancellation are obtained. We evaluate the outage over different splits of the elements and varying pathloss differences between the two \acp{UE}. The analysis shows that for small pathloss differences, the split should be chosen such that most of the \ac{IRS} elements are configured to boost the stronger \ac{UE}, while for large pathloss differences, it is more beneficial to boost the weaker \ac{UE}. Finally, we investigate a robust selection of the elements' split under the criterion of minimizing the maximum outage between the two \acp{UE}.
\end{abstract}

%\begin{IEEEkeywords}
%	\ac{IRS}, \ac{NOMA}, Gamma moments matching, outage analysis, interference cancellation. 
%\end{IEEEkeywords}

\IEEEpeerreviewmaketitle

\section{Introduction}
\Acfp{IRS}, also known as \acfp{RIS}, have been identified as a promising technology to enhance the spectral and energy efficiency of \ac{B5G} wireless systems \cite{Renzo19, Wu19}. Those surfaces consist of a large number of low-cost reconfigurable elements whose electromagnetic response to impinging/incident waves can be modified. Phase adjustment of the waves across the different elements allows the surface to perform passive beamforming, which is beneficial in the context of extending the coverage area, focusing the energy towards a certain \acf{UE}, reducing interference, and more \cite{Wu20, Basar19}. 

Another technology that has gained interest over the past couple of years is \acf{NOMA}. With \ac{NOMA}, the \acp{UE} can contest the same time-frequency resources in a non-orthogonal manner, which may lead to a higher spectral efficiency, lower access latency, improved user fairness, etc \cite{Dai18, Ding17}. The combination of \ac{NOMA} with \acp{IRS} has gained attention recently, with many works showing potential gains in terms of energy efficiency, sum-rate, and outage performance \cite{Ding20a, Fu19, Yang20, Cheng20, Ding20b}. An important aspect is how to configure the elements of the surface. In some works, the phase shifts across the different elements are set jointly according to a certain design criterion \cite{Yang20, Fu19, Zeng20}, such as maximizing the sum-rate. Other works consider the case where the entire surface is used to boost one of the \ac{NOMA} \acp{UE} \cite{Ding20a, Cheng20}.

We consider in this paper an \ac{IRS}-assisted \ac{NOMA} uplink, in which the elements of the \ac{IRS} are split between the two \ac{NOMA} \acp{UE}, i.e., part of the surface is used to coherently combine the signal of the first \ac{UE}, while the other part is used to coherently combine the signal of the second one. We assume the communication to take place primarily through the surface, e.g., due to blockage of the direct links to the \ac{BS}. All the links are assumed to undergo Nakagami-$m$ fading, allowing to flexibly capture \acf{LOS} and \acf{NLOS} propagation conditions. We analyze the outage probability under \ac{NOMA} \acf{IC} for different splits of the elements and pathloss differences between the \acp{UE}. To obtain tractable expressions, we approximate the received powers of the \acp{UE} as Gamma \acfp{RV} in a fashion similar to \cite{Tahir20cPre, Lyu20} via second-order moments matching. We finally analyze an outage-robust selection of the elements' split between the two \acp{UE} and discuss its impact on the performance limits of such a system. 

\section{System Model}
\begin{figure}
	\centering
	\begin{center}
		\begin{tikzpicture}
		\node[,] (image) at (0,0) {\includegraphics[width=0.8\linewidth]{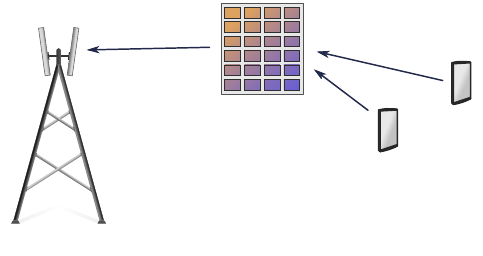}};
		\node[,] at (-2.68, -1.6) {\legendfont BS};
		\node[,] at (0.35, 0.27) {\legendfont IRS};
		\node[,] at (2.2, -0.55) {\legendfont UE1};
		\node[,] at (3.25, 0.14) {\legendfont UE2};
		\end{tikzpicture}
		\vspace{-3mm}
	\end{center}
	\caption{The uplink \ac{IRS}-assisted \ac{NOMA} setup.}
\label{fig:00}
\vspace{-2mm}
\end{figure}
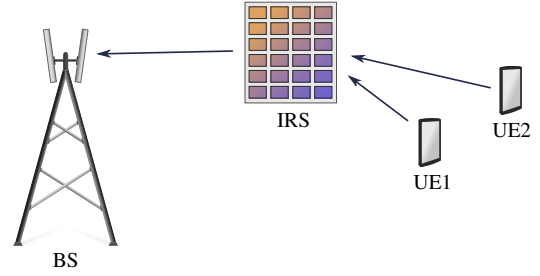
We consider a single-antenna two-\ac{UE} \ac{NOMA} uplink assisted by an $N$-elements \ac{IRS}, as shown in \Cref{fig:00}. At the \ac{BS}, the received signal is given by
\begin{align}\label{eq:1}
r = \sum_{i = 1}^{2} \sqrt{\ell_{\mathrm{BS}} \ell_{h_i}P_i}\, \mathbf{h}_{\mathrm{BS}}^T\mathbf{\Phi}\,\mathbf{h}_i \,x_i + w,
\end{align}
where $\mathbf{h}_{\mathrm{BS}} \in \mathbb{C}^{N}$, and $\mathbf{h}_i \in \mathbb{C}^{N}$ are the small-scale fading coefficients of the \ac{BS}-\ac{IRS} and \ac{IRS}-\ac{UE} links, respectively. The parameters $\ell_{\mathrm{BS}}$ and $\ell_{h_i}$ are the corresponding pathlosses, $P_i$ and $x_i$ are the transmit power and signal of the $i^{\mathrm{th}}$-\ac{UE}, and $w$ is the zero-mean Gaussian noise with power $P_w$. The phase-shift matrix $\mathbf{\Phi} \in \mathbb{C}^{N \times N}$ is defined as
\begin{align}\label{eq:2}
	\mathbf{\Phi} = \mathrm{diag} \big(e^{j\phi_1}, e^{j\phi_2}, \dots, e^{j\phi_N}\big),
\end{align}
where $\phi_n$ is the phase-shift applied at the $n^{\textrm{th}}$-element of the \ac{IRS}. Note that the \ac{IRS} term can be written equivalently as
\begin{align}\label{eq:3}
	\mathbf{h}_{\mathrm{BS}}^T\mathbf{\Phi}\,\mathbf{h}_i = \sum_{n = 1}^{N} e^{j\phi_n} \mathbf{h}_{\mathrm{BS}, n}  \,  \mathbf{h}_{i, n},
\end{align}
where $\mathbf{h}_{\mathrm{BS}, n}$ and $\mathbf{h}_{i, n}$ are the $n^{\textrm{th}}$-elements of $\mathbf{h}_{\mathrm{BS}}$ and  $\mathbf{h}_{i}$, respectively. To be flexible in terms of modeling \ac{LOS} and \ac{NLOS} propagation conditions, the links are assumed to undergo Nakagami-$m$ fading, i.e.,
\begin{align}
\begin{split}
	|\mathbf{h}_{\mathrm{BS}, n}| &\sim \mathrm{Nakagami} (m_{\mathrm{BS}}, 1), \\
	|\mathbf{h}_{i, n}| &\sim \mathrm{Nakagami} (m_{h_i}, 1),
\end{split}	
\end{align}
where $m_{\mathrm{BS}}$ and $m_{h_i}$ are the corresponding distribution parameters.

In this work, we consider the case where the elements of the \ac{IRS} is split between the two \acp{UE}, i.e., a total of $N_1$ elements are configured to coherently combine the signal of \ac{UE}1, while $N_2 = N - N_1$ elements are configured for \ac{UE}2. The phases are then set to
\begin{align}\label{eq:6}
\phi_{n} = - \arg{\big(\mathbf{h}_{\mathrm{BS}, n} \, \mathbf{h}_{i, n}\big)}, \quad n \in \mathcal{C}_i,
\end{align}
where $\mathcal{C}_i$ is the set of elements that are configured to boost the $i^{\textrm{th}}$-\ac{UE}. Therefore, the \ac{IRS} term can be written as
\begin{align}\label{eq:split}
	\mathbf{h}_{\mathrm{BS}}^T\mathbf{\Phi}\,\mathbf{h}_i = \underbrace{\sum_{n \in \mathcal{C}_i} \la \mathbf{h}_{\mathrm{BS}, n} \ra \, \la \mathbf{h}_{i, n} \ra}_{\substack{\textrm{coherently combined} \\ \textrm{part of the } i^{\textrm{th}}\textrm{-UE}}} + \underbrace{\sum_{n \in \overline{\mathcal{C}}_i} e^{j\phi_n} \mathbf{h}_{\mathrm{BS}, n} \, \mathbf{h}_{i, n}}_{\substack{\textrm{randomly combined} \\ \textrm{part of the } i^{\textrm{th}}\textrm{-UE}}},
\end{align}
where the complement set $\overline{\mathcal{C}}_i$ is the set of elements that are not configured for the $i^{\textrm{th}}$-\ac{UE}, and thus will result in a random combining of its phases. Note that $\mathcal{C}_1 = \overline{\mathcal{C}}_2$, i.e., the part that will coherently combine the signal of one of the \acp{UE}, will randomly combine the signal of the other one. This is under the assumption that the channels of the two \acp{UE} are  uncorrelated. 

Since we apply the Gamma moment matching often in this work, we state how it is performed in the following lemma.

\begin{lemma}\label{lemma:1}
	Let $X$ be a non-negative \ac{RV} with first and second moments given by $\mu_X = \mathbb{E}\{X\}$ and $\mu_X^{(2)} = \mathbb{E}\{X^2\}$, respectively. The Gamma \ac{RV} $Y \sim \Gamma(k, \theta)$ with the same first and second moments has shape $k$ and scale $\theta$ parameters
	\begin{align*}
	k = \frac{\mu_X^2}{\mu_X^{(2)} - \mu_X^2}, \quad \quad \theta = \frac{\mu_X^{(2)} - \mu_X^2}{\mu_X}.
	\end{align*}		
\end{lemma}
Additionally, Gamma \acp{RV} satisfy the scaling property, in the sense that if $Y \sim \Gamma(k, \theta)$, then $cY \sim \Gamma(k, c\theta)$. 

\section{Outage Analysis}
Evaluating metrics such as the \ac{NOMA} outage probability requires an access to the statistics of the channel output power, preferably the full distribution. As can be seen in \eqref{eq:split}, that is no easy task. Therefore, we resort to approximations. Our approach here is to approximate the received powers of the \ac{NOMA} \acp{UE} as Gamma \acp{RV}. On the one hand, the Gamma distribution can accurately model the power of many fading distributions, and on the other hand, it allows for tractable expressions when evaluating the outage, as we will see later.

\subsection{Statistics of the Received Power}

Let $Z_i$ be the channel power of the $i^{\textrm{th}}$-\ac{UE}, i.e., 
\begin{align}\label{eq:zi}
	Z_i = \ell_{\mathrm{BS}} \ell_{h_i} \bigg| \sum_{n \in \mathcal{C}_i} \la \mathbf{h}_{\mathrm{BS}, n} \ra \, \la \mathbf{h}_{i, n} \ra + \sum_{n \in \overline{\mathcal{C}}_i} e^{j\phi_n} \mathbf{h}_{\mathrm{BS}, n} \, \mathbf{h}_{i, n}\bigg|^2.
\end{align}
Our goal here is to approximate $Z_i$ as a Gamma \ac{RV} via second-order moments matching, which requires an access to its first two moments. To simplify matters, we first address the statistics of the two sum terms inside.

\begin{lemma} For the two sum terms in \eqref{eq:zi} given by
	\begin{align*}
		S_{\mathcal{C}_i} &= \sum_{n \in \mathcal{C}_i} \la \mathbf{h}_{\mathrm{BS}, n} \ra \, \la \mathbf{h}_{i, n} \ra, \\
		S_{\overline{\mathcal{C}}_i} &= \sum_{n \in \overline{\mathcal{C}}_i} e^{j\phi_n} \mathbf{h}_{\mathrm{BS}, n} \, \mathbf{h}_{i, n},
	\end{align*}
	their distributions are approximated as
	\begin{align*}
		S_{\mathcal{C}_i} & \stackrel{\mathrm{approx}}{\sim}  \Gamma\Bigg(N_i\frac{\mu_i^2}{1 - \mu_i^2}\,, \, \frac{1 - \mu_i^2}{\mu_i} \Bigg), \\
		S_{\overline{\mathcal{C}}_i} & \stackrel{\mathrm{approx}}{\sim}  \mathcal{CN}\big(0, N - N_i\big),
	\end{align*}
%	\begin{align*}
%		S_{\mathcal{C}_i} & \stackrel{\mathrm{approx}}{\sim}  \Gamma\big(N_i k_S, \, \theta_S \big)~, \\
%		S_{\overline{\mathcal{C}}_i} & \stackrel{\mathrm{approx}}{\sim}  \mathcal{CN}\big(0,\, N - N_i\big)~.
%	\end{align*}
%	where
%	\begin{align*}
%		k_S = \frac{\mu_i^2}{1 - \mu_i^2}~, \quad \theta_S = \frac{1 - \mu_i^2}{\mu_i}~,
%	\end{align*}
	with 
	\begin{align*}
	\mu_i = \mathbb{E}\big\{| \mathbf{h}_{\mathrm{BS}, n} | \, | \mathbf{h}_{i, n} |\big\} =  \frac{\Gamma(m_{\mathrm{BS}} + \frac{1}{2}) \Gamma(m_{h_i} + \frac{1}{2})}{\Gamma(m_{\mathrm{BS}})\Gamma(m_{h_i}) (m_{\mathrm{BS}}\, m_{h_i})^{1/2}},
	\end{align*}
	being the mean of the product of two independent Nakagami \acp{RV}.
\end{lemma}
\begin{proof}
	We follow a similar approach as in \cite{Tahir20cPre}. The first term is a sum of $N_i$ i.i.d. positive \acp{RV} that can be well approximated by a Gamma \ac{RV} via moments matching (\Cref{lemma:1}), requiring access only to the mean of the composite channel ($\mu_i$). For the second term, it is a sum of $N - N_i$ complex-valued i.i.d. \acp{RV}, which are approximated by a zero-mean complex Gaussian.
\end{proof}
The channel power now can be compactly written as 
\begin{align}\label{eq:09}
	Z_i = \ell_{\mathrm{BS}} \ell_{h_i} | S_{\mathcal{C}_i} + S_{\overline{\mathcal{C}}_i} |^2,
\end{align}
with the first two moments given by the following lemma.
\begin{lemma}\label{lemma:4}
	The first two moments of the channel power under elements splitting are given by
	\begin{align*}
	\mu_{Z_i}^{\vphantom{(2)}} &= \ell_{\mathrm{BS}} \ell_{h_i} \Big(\mu_{S_{\mathcal{C}_i}}^{(2)} + \mu_{|S_{\overline{\mathcal{C}}_i}|}^{(2)}\Big), \\
	\mu_{Z_i}^{(2)} &= (\ell_{\mathrm{BS}} \ell_{h_i})^2 \Big(\mu_{S_{\mathcal{C}_i}}^{(4)} + \mu_{|S_{\overline{\mathcal{C}}_i}|}^{(4)} +  4\,\mu_{S_{\mathcal{C}_i}}^{(2)}\,\mu_{|S_{\overline{\mathcal{C}}_i}|}^{(2)}\Big),
	\end{align*}
	where
	\begin{align*}
	\mu_{ S_{\mathcal{C}_i}}^{(p)}  &= \frac{\Gamma \Big(N_i \frac{\mu_i^2}{1 - \mu_i^2}  + p \Big) \Big( \frac{1 - \mu_i^2}{\mu_i}\Big)^p }{\Gamma \Big(N_i  \frac{\mu_i^2}{1 - \mu_i^2} \Big)}, \\  
	\mu_{|S_{\overline{\mathcal{C}}_i}|}^{(p)}  &= \Gamma\Big(1 + \frac{p}{2}\Big) {(N - N_i)}^{p/2}.
	\end{align*}
\end{lemma}
\begin{proof}
	Expanding \eqref{eq:09}, we have
	\begin{align*}
		\mu_{Z_i}^{\vphantom{(2)}} = \ell_{\mathrm{BS}} \ell_{h_i}\mathbb{E}\big\{S_{\mathcal{C}_i}^2 +  |S_{\overline{\mathcal{C}}_i}|^2 + 2 S_{\mathcal{C}_i} \Re\{S_{\overline{\mathcal{C}}_i}\} \big\}.
	\end{align*}
	We make the assumption here that the phase of $S_{\overline{\mathcal{C}}_i}$ is zero-mean symmetric. This is valid since it results from an out-of-phase summation of the terms. Therefore, its real part will be zero-mean as well, leading to $\mathbb{E}\big\{ S_{\mathcal{C}_i} \Re\{S_{\overline{\mathcal{C}}_i}\} \big\} = \mathbb{E}\big\{ S_{\mathcal{C}_i} \big\} \mathbb{E}\big\{\Re\{S_{\overline{\mathcal{C}}_i}\}\big\} = 0$, giving the final result.
	We proceed in a similar manner for $\mu_{Z_i}^{(2)}$. After the expansion we get
	\begin{align*}
		\begin{split}
		\mu_{Z_i}^{(2)} = (\ell_{\mathrm{BS}} \ell_{h_i})^2& \mathbb{E}\big\{ \Sc^4 + |\Scb|^4 + 2\Sc^2|\Scb|^2 + 4\Sc^3\Re\{\Scb\} \\
		&\quad~ + 4\Sc|\Scb|^2\Re\{\Scb\} + 4 \Sc^2\Re\{\Scb\}^2 \big\}. \\
		\end{split}	
	\end{align*}
	Following the assumptions of independence and zero-mean symmetry, we have
	\begin{align*}
		\mathbb{E}\big\{\Sc^3\Re\{\Scb\} \big\} &= \mathbb{E}\big\{\Sc^3\big\} \mathbb{E}\big\{\Re\{\Scb\}\big\} = 0,
	\end{align*}
	and
	\begin{align*}
	\mathbb{E}\big\{\Sc|\Scb|^2\Re\{\Scb\} \big\} &= \mathbb{E}\big\{\Sc\big\}\mathbb{E}\big\{ |\Scb|^2\Re\{\Scb\} \big\} \\
	&= \mathbb{E}\big\{\Sc\big\}\mathbb{E}\big\{\Re\{\Scb\}^3  \\ 
	&\quad\quad\quad\quad\quad+ \Im\{\Scb\}\Re\{\Scb\} \big\} \\
	&= 0,
	\end{align*}
	where the final result follows from the fact that the third moment is zero as well (due to symmetry), and independence between the real and imaginary parts.
	For the last term, we assume that the power is split equally across the real and imaginary parts, and therefore
	\begin{align*}
		\mathbb{E}\big\{\Sc^2\Re\{\Scb\}^2\big\} = \mathbb{E}\big\{\Sc^2\big\} \mathbb{E}\big\{|\Scb|^2\big\}/2.
	\end{align*}
	We get the final results by collecting the terms back and substituting the moments of Gamma and Rayleigh (magnitude of Gaussian) \acp{RV}.
\end{proof}
After scaling with the transmit power, the received power of the $i^{\textrm{th}}$-\ac{UE} is given by
	\begin{align}
	Z_iP_i \stackrel{\mathrm{approx}}{\sim}  \Gamma\big(k_i, P_i\theta_i\big),
	\end{align}
where $k_i$ and $\theta_i$ are the Gamma parameters matched to the moments in \Cref{lemma:4}.

\subsection{Outage Probability under Interference Cancellation}
Before applying \ac{IC}, the \ac{SINR} outage of the $i^{\textrm{th}}$-\ac{UE} under the presence of interference from the $j^{\textrm{th}}$-\ac{UE} is defined as
\begin{align}\label{eq:15}
p_{\mathrm{out}}^{(i)} = \mathbb{P}\bigg\{ \frac{Z_i P_i}{Z_j P_j + P_w } \leq \epsilon \bigg\},
\end{align}
where $\epsilon$ is the outage threshold. Under the Gamma approximation, this can be calculated as \cite{Tahir20cPre}
\begin{align*}
\begin{split}
p_{\mathrm{out}}^{(i)}  \approx & ~ I\bigg(\frac{\epsilon \hat{\theta}_j}{\hat{\theta}_i + \epsilon \hat{\theta}_j};\,\hat{k}_i,\,\hat{k}_j \bigg),
\end{split}
\end{align*}
where $I(.;.,.)$ is the regularized incomplete beta function, and
\begin{align}
\begin{split}
\hat{k}_i = k_i, \quad \quad &\hat{\theta}_i = \theta_iP_i,\\
\hat{k}_j = \frac{\big(k_j\theta_jP_j + P_w\big)^2}{k_j (\theta_jP_j)^2}, \quad \quad &\hat{\theta}_j = \frac{k_j (\theta_jP_j)^2}{k_j\theta_jP_j + P_w}.
\end{split}
\end{align}
If the interference is later removed via \ac{IC}, only the noise remains; we define the \ac{SNR} outage as
\begin{align}\label{eq:16}
p_{\mathrm{out,\,SNR}}^{(i)} = \mathbb{P}\bigg\{ \frac{Z_i P_i}{P_w } \leq \epsilon \bigg\},
\end{align}
which is simply the \ac{CDF} of $Z_i$ (Gamma \ac{RV}) evaluated at $\epsilon P_w / P_i$. We consider here a parallel \ac{IC} scheme, in which \ac{UE}1, \ac{UE}2, or both can be detected correctly at the first iteration if they are above the outage threshold. Whatever remains could be detected in the next iteration after \ac{IC} in the presence of noise only. Following the path of successful detection for the $i^{\textrm{th}}$-\ac{UE}, its outage probability under \ac{IC} is given by  \cite{Tahir20cPre}

\begin{align}\label{eq:1333}
p_{\mathrm{out,\,IC}}^{(i)} \approx 1 - \min \big(p_{\mathrm{succ}}^{(i)} + p_{\mathrm{succ}}^{(j)} \, p_{\mathrm{succ,\,SNR}}^{(i)} \,,\, p_{\mathrm{succ,\,SNR}}^{(i)} \big),
\end{align}
where $p_{\mathrm{succ}}^{(i)} = 1 - p_{\mathrm{out}}^{(i)}$ and $p_{\mathrm{succ,\,SNR}}^{(i)} = 1 - p_{\mathrm{out,\,SNR}}^{(i)}$ are the corresponding success probabilities. 

\section{Analysis of an Example Scenario}
%\vspace{-0.5mm}
\begin{reply}We consider a scenario where the communication between the \ac{NOMA} \acp{UE} and the \ac{BS} takes place through a 32-elements \ac{IRS}, and evaluate the outage performance using \eqref{eq:1333}\end{reply}. We assume the \ac{IRS} to have a strong \ac{LOS} connection to the \ac{BS}, while the \acp{UE} have moderate \ac{LOS} to the \ac{IRS}, with \ac{UE}1 having a stronger \ac{LOS} than \ac{UE}2. This is set by adjusting the corresponding Nakagami $m$ parameters. The pathloss of \ac{UE}1 is fixed to $-70$\,dB, while for \ac{UE}2, it varies. \begin{reply}
Without loss of generality, we assume that both \acp{UE} are transmitting with the same power, i.e., $P_1 = P_2$. In practice, the \acp{UE} might transmit with different powers; however, that does not affect the validity of our analysis here. It holds for any choice of $P_1$ and $P_2$. The simulation parameters are summarized in \Cref{table:1}. \end{reply}

\begin{table}
	\begin{center}
%		\fontsize{10pt}{12pt}\selectfont
%		\begin{tabularx}{0.95\linewidth}{l |l}
		\fontsize{8pt}{10pt}\selectfont
		\begin{tabularx}{0.8\linewidth}{l |l}
			\hline
			\textbf{Parameter} & \textbf{Value} \\ \specialrule{1pt}{0pt}{1pt}
			\#IRS elements & $N = 32$ \\
			Transmit powers & $P_1 = P_2 = $ $30$\,dBm \\
			Nakagami parameters & \parbox[t]{5cm}{$m_{\mathrm{BS}} = 6$\\
				$m_{h_1} = 3$, $m_{h_2} = 1.5$\vspace{1mm}} \\ 
			Pathlosses & \parbox[t]{5cm}{$\ell_{\mathrm{BS}} = -65$\,dB	\\
				$\ell_{h_1} = -70$\,dB, $\ell_{h_2}$ is variable \vspace{1mm}}\\
			Noise power & $P_w = -110$\,dBm \\
			\hline
		\end{tabularx}
	\end{center}
	\vspace{-2mm}
	\caption{Scenario parameters.}
	\label{table:1}
	\vspace{-5mm}
\end{table}
%\vspace{-0.5mm}
\subsection{Impact of Elements Splitting}
%\vspace{-0.5mm}
We define the split factor $\alpha$ as the percentage of elements that are allocated for the coherent combining of the signal of \ac{UE}1. Given that $N_1$ elements are allocated to \ac{UE}1, the split factor then is defined as
\begin{align}
	\alpha = N_1 / N,
\end{align}
and thus the number of elements allocated for \ac{UE}2 is
\begin{align}
	N_2 = N - \ceil*{\alpha N}.
\end{align}
When $\alpha =1$, all the elements are allocated to \ac{UE}1, while for $\alpha = 0$, all the elements are allocated to \ac{UE}2, etc.

We investigate the outage performance over the split factor $\alpha$ for different outage thresholds. \Cref{fig:01} shows the performance when the \acp{UE} have the same pathloss, and therefore the average power gap between them is $0$\,dB (ignoring the surface processing). In that case, the outage probability for both \acp{UE} is minimized, if most of the elements are configured to boost one of the \acp{UE}. The reason for this is, when the split is close to $50\%$, then it is more likely that both \acp{UE} will be received with similar strength at the \ac{BS}. This, in turn, makes the \ac{IC} more difficult, since the \acp{UE} would suffer strong interference from each other. Therefore, when the pathloss difference between the two \acp{UE} is small, it makes sense to focus on boosting one of the \acp{UE}, such that the power gap between them increases, allowing the stronger \ac{UE} to be detected correctly at the first \ac{IC} iteration with high probability.

\Cref{fig:02} and \Cref{fig:03} show the outage performance when the pathloss of \ac{UE}2 is $5$\,dB and $10$\,dB higher than \ac{UE}1, respectively. We observe that as the gap increases, and at low outage thresholds, the split moves towards boosting \ac{UE}2. In this case, the two \acp{UE} have a natural power gap due to the pathloss difference, and therefore the \ac{IRS} can be used to enhance the performance of the weaker user (\ac{UE}2). It can also be observed that as the gap increases, better performance is achieved for both \ac{UE}s. This indicates that when it comes to \ac{NOMA} user pairing, the \ac{BS} should avoid pairing users with similar pathlosses. However, the weak \ac{UE} should be strong enough such that after the combining at the surface, it is able to overcome the noise at the \ac{BS} receiver. Regarding the accuracy of the analysis, we observe that the approximations hold well for the strong \ac{UE}. As for the weak \ac{UE}, and at low outage thresholds, a relatively large gap exists between analysis and simulation for some values of the split factor, suggesting that the Gamma approximation does not hold well under such splitting conditions.

\begin{figure}
\captionsetup{textfont={normalsize}, labelfont={normalsize}, farskip=15pt, captionskip=2pt}

\subfloat[Equal pathloss case.]{
	\resizebox{1\linewidth}{!}{%
		\pgfplotsset{width=250pt, height=200pt, compat = 1.9}
		\input{graphics/fig_0dB}
	}
	\label{fig:01}
}\
\subfloat[\ac{UE}2 $5$\,dB weaker.]{
	\resizebox{1\linewidth}{!}{%
		\pgfplotsset{width=250pt, height=200pt, compat = 1.9}
		\input{graphics/fig_5dB}
	}
	\label{fig:02}
}\
\subfloat[\ac{UE}2 $10$\,dB weaker.]{
	\resizebox{1\linewidth}{!}{%
		\pgfplotsset{width=250pt, height=200pt, compat = 1.9}
		\input{graphics/fig_10dB}
	}
	\label{fig:03}
}\
\vspace{2mm}
\caption{Analysis of the outage probability vs. the split factor for different outage thresholds.}\label{fig:animals}
\end{figure}
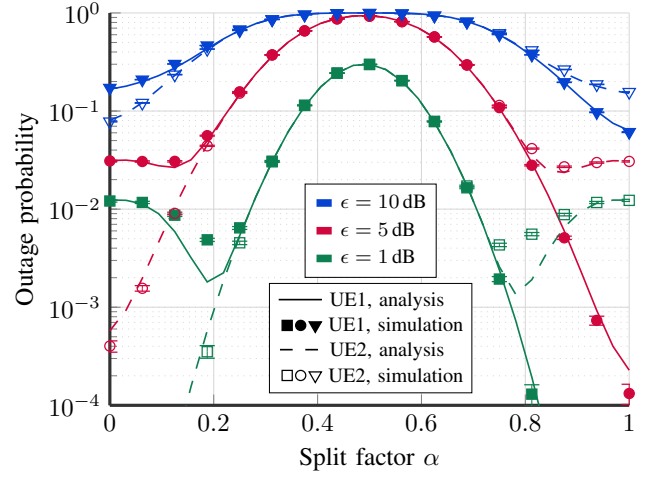
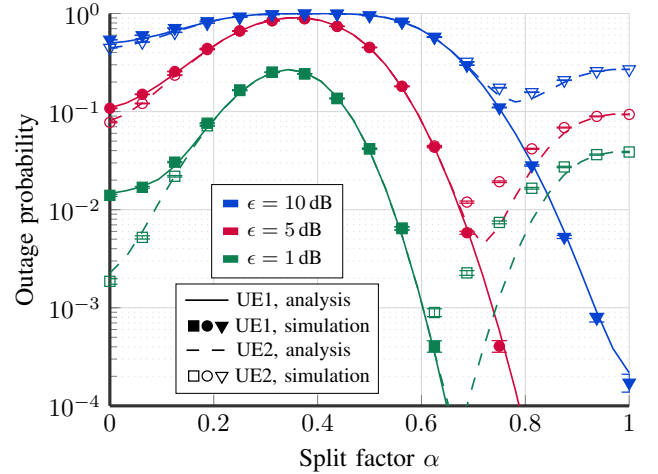
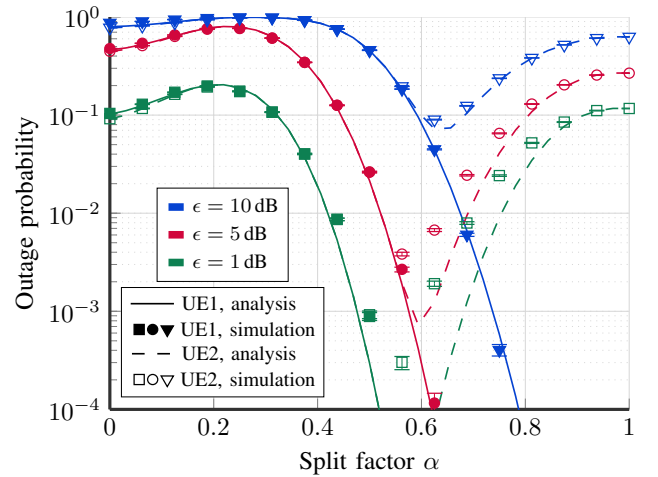

\subsection{Selection of the Split Factor}
We consider the selection of the split factor $\alpha$ from a robust perspective. To ensure that boosting the performance of one \ac{UE} does not come at the cost of degrading the performance of the other one, the split factor is chosen according to
\begin{align}\label{eq:argmin}
	\alpha_{\textrm{robust}} = \argmin_\alpha \, \max_i \,p_{\mathrm{out,\,IC}}^{(i)}.
\end{align}
In Figure \ref{fig:01} to \ref{fig:03}, this would correspond to the points where the \ac{UE}2 outage diverges from \ac{UE}1 and starts saturating (on the right side). However, at high outage thresholds, it can be observed that the outage probability of \ac{UE}2 is very high, no matter what split is applied. For that reason, we introduce the notion of a limiting threshold $\lambda$. If the weak \ac{UE} outage probability is higher than $\lambda$, then the entire \ac{IRS} is used to boost the strong \ac{UE}, as allocating elements to the weak \ac{UE} would be a waste of the surface elements. Assuming \ac{UE}2 is the weaker \ac{UE}, \eqref{eq:argmin} is modified as follows
\begin{align}\label{eq:argmin2}
	\alpha_{\textrm{robust}} =
	\begin{cases} 
		\argmin_\alpha \, \max_i \,p_{\mathrm{out,\,IC}}^{(i)}, &\mbox{if }  p_{\mathrm{out,\,IC}}^{(2)} < \lambda, \\
		1, & \mbox{otherwise}.
	\end{cases}
\end{align}
Although the analysis results shown in the previous subsection is not accurate at low outage thresholds, it can be seen that the robust point occurs almost at the same $\alpha$ for both analysis and simulation. This motivates the use of the analysis as a method to determine $\alpha_{\textrm{robust}}$. Solving \eqref{eq:argmin2} in closed-form is difficult due to the complexity of the functions involved. We thus rely on performing a search for determining the optimal point. Recall that $\alpha = N_1 / N$ with $N_1 = 1,2, \dots, N$ (i.e., the maximum number of possibilities is $N$), meaning that the search can be performed quickly.

\begin{figure}[t]
	\centering
	\resizebox{1\linewidth}{!}{%
				\pgfplotsset{width=250pt, height=200pt, compat = 1.9}
				\input{graphics/fig_robust}
	}
	\caption{Robust selection of the split factor for different pathloss gaps (search via analysis vs. exhaustive simulations).}
	\label{fig:04}
\end{figure}
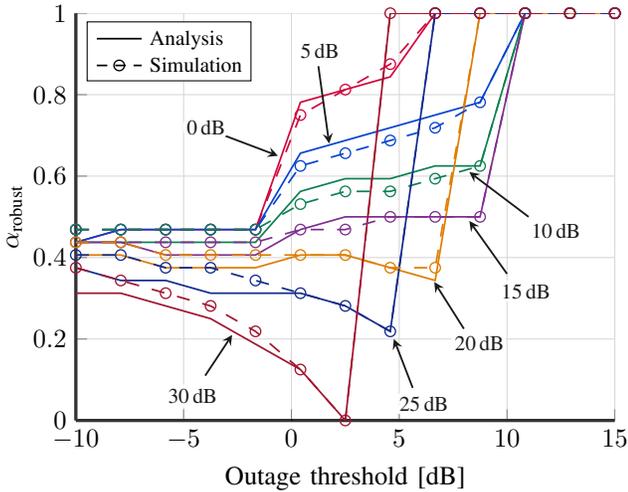

\Cref{fig:04} shows $\alpha_{\textrm{robust}}$ obtained by search through exhaustive simulations vs. analysis for different pathloss gaps between the two \acp{UE}. We observe that at low pathloss gaps, the split is chosen to boost \ac{UE}1 (the stronger \ac{UE}), since it improves the performance of the \ac{NOMA} \ac{IC}. As the pathloss of \ac{UE}2 increases, the robust \ac{IRS} strategy attempts to compensate for the high pathloss by allocating more elements to \ac{UE}2. The sudden jumps to $1$ are due to the limiting threshold, which is set to $\lambda = 10^{-1}$ here. This indicates that at those outage thresholds, the performance of \ac{UE}2 is unacceptable anyway, that it is better to use the entire surface to boost \ac{UE}1. Also, at low outage thresholds, we observe that a split close to $50\%$ seems to be the robust selection, while at high outage thresholds, the selection across the different gaps can vary substantially. Note that at low outage thresholds, the outage probability is very low, which makes the simulation-based selection inaccurate, since it would require a huge number of simulation samples. This is the advantage of the analytical based approach, since it can predict the performance, even at very low outage probabilities.

\section{Conclusion}
In this work, we investigate the outage performance of a two-\acp{UE} \ac{IRS}-assisted \ac{NOMA} uplink, in which the elements of the surface are split between the two \acp{UE}. A Gamma approximation of the \acp{UE} received power is applied, allowing for tractable expressions of the outage under \ac{NOMA} \ac{IC}, while being able to capture \ac{LOS} and \ac{NLOS} propagation conditions. Our results show that when the pathloss difference between the \ac{NOMA} \acp{UE} is small, then better outage performance is achieved if most of the surface elements are configured to boost the stronger \ac{UE}. This further increases the power gap between the two \acp{UE} at the \ac{BS} receiver, and thus improves the performance of \ac{IC} at the first iteration. As the pathloss difference increases, then more elements should be allocated to boosting the weaker \ac{UE}. In this case, a natural power gap exists between the \acp{UE} due to the pathloss difference, and therefore the split should be chosen to boost the performance of the weak \ac{UE}, such that it is able to overcome the noise. At the end, we investigate a robust selection of the split factor based on minimizing the maximum outage between the two \acp{UE}, and show that such a split can be well predicted using our analysis.

%\vspace{-2mm}
\input{./Acronyms.tex}
\bibliographystyle{IEEEtran}
\bibliography{IEEEabrv,./LocalRefs}

\end{document}

%% file: tikzsetup.tex
\newcommand{\legendfont}{\fontsize{8pt}{8pt}\selectfont}
\newcommand{\legendfontt}{\fontsize{9pt}{9pt}\selectfont}
\newcommand{\labelfont}{\fontsize{10pt}{10pt}\selectfont}

\definecolor{c1}{RGB}{215, 196, 255}
\definecolor{c2}{RGB}{89,51,84}
\definecolor{cgrid}{RGB}{215,215,215}
\definecolor{cgrid2}{RGB}{215,215,215}
\definecolor{dark01}{RGB}{50, 50, 50}
\definecolor{dark011}{RGB}{25, 25, 25}
\definecolor{greenC1}{RGB}{13, 128, 82}
\definecolor{redC1}{RGB}{209, 6, 60}
\definecolor{blueC1}{RGB}{6, 67, 209}
\definecolor{c1}{RGB}{209, 6, 60}
\definecolor{c2}{RGB}{6, 67, 209}
\definecolor{c3}{RGB}{13, 128, 82}
\definecolor{c4}{RGB}{126, 47, 142}
\definecolor{c5}{RGB}{222, 125, 0}
\definecolor{c6}{RGB}{20, 43, 140}
\definecolor{c7}{RGB}{162, 20, 47}

\tikzset{
	gridc/.style= {dotted, cgrid2},
	linew/.style= {line width=0.65pt},
	linew2/.style= {line width=0.75pt},
	marksz/.style= {mark options={scale=1, fill opacity=1, solid, line width=0.5}},
	marksz2/.style= {mark options={scale=1.4, fill opacity=1, solid, rotate=180, line width=0.5}},
	marksz3/.style= {mark options={scale=1, fill=white, fill opacity=0.5, solid, line width=0.5}},
	errormarkSty/.style = {rotate=90, mark size=3pt, solid, line width=0.5pt},
	invisible/.style={opacity=0},
	visible on/.style={alt={#1{}{invisible}}},
	alt/.code args={<#1>#2#3}{%
		\alt<#1>{\pgfkeysalso{#2}}{\pgfkeysalso{#3}} % \pgfkeysalso doesn't change the path
	},
}

\pgfplotsset{
	axisSetup/.style= {axis x line=bottom, axis y line=left, tick align=inside, axis line style={-, line width=1.25pt, color=dark01}},
	short Legend0/.style={%
		legend image code/.code={
			\draw[##1,line width=3pt] plot coordinates {(0pt,0pt) (6pt,0pt)};
		}
	},
	short Legend1/.style={%
		legend image code/.code={
			\draw[##1,linew] plot coordinates {(0pt,0pt) (9pt,0pt)};
			\draw[##1,linew, dashed] plot coordinates {(0pt,2pt) (9pt,2pt)};
			\draw[##1,mark=square, marksz, linew] plot coordinates {(-4pt,1pt)};
		}
	},
	short Legend2/.style={%
		legend image code/.code={
			\draw[##1,linew] plot coordinates {(0pt,0pt) (9pt,0pt)};
			\draw[##1,linew, dashed] plot coordinates {(0pt,2pt) (9pt,2pt)};
			\draw[##1,mark=o, marksz, linew] plot coordinates {(-4pt,1pt)};
		}
	},
	short Legend3/.style={%
		legend image code/.code={
			\draw[##1,linew] plot coordinates {(0pt,0pt) (9pt,0pt)};
			\draw[##1,linew, dashed] plot coordinates {(0pt,2pt) (9pt,2pt)};
			\draw[##1,mark=triangle, linew, marksz2] plot coordinates {(-4pt,1pt)};
		}
	},
	ylabel right/.style={
		after end axis/.append code={
			\node [rotate=90, anchor=north] at (rel axis cs:1,0.5) {#1};
		}   
	}
}

%% file: graphics/fig_0dB.tex
\begin{tikzpicture}
	\begin{semilogyaxis}[
	xlabel={Split factor $\alpha$},
	ylabel={Outage probability},
	label style={font=\labelfont},
	ylabel shift = -1mm,
	ylabel right ={~},
%	xlabel shift = -1mm,	
	ymin=0.0001, ymax=1,
	xmin=0, xmax=1,
	ytick pos=left,
	axisSetup,
	ymajorgrids=true,
	xmajorgrids=true,
	yminorgrids=true,
	xminorgrids=true,
	major x grid style={solid, cgrid},
	major y grid style={solid, cgrid},
	minor x grid style={solid, cgrid},
	minor y grid style={gridc},
	legend style={font=\legendfont, name=legendNode, at={(0.375,0.445)},anchor=west},
	legend cell align=left,
%	legend pos=south west,
	]
	% Solid
	\foreach \i/\c in {1/greenC1, 5/redC1, 9/blueC1}{
		\edef\temp{\noexpand \addplot [linew, color=\c, mark=none, forget plot] table [x=x\i, y=y\i, col sep=comma] {graphics/results/fig_0dB.csv};}
		\temp
	};
	% Dashed
	\foreach \i/\c in {2/greenC1, 6/redC1, 10/blueC1}{
		\edef\temp{\noexpand \addplot [linew, color=\c, mark=none, dashed, dash pattern=on 5pt off 5pt, forget plot] table [x=x\i, y=y\i, col sep=comma] {graphics/results/fig_0dB.csv};}
		\temp
	};
	% Error bars
	\foreach \i/\c/\m in {3/greenC1/square*, 4/greenC1/square, 7/redC1/*, 8/redC1/o}{
		\edef\temp{\noexpand \addplot [error bars/.cd, y dir=both, y explicit, error mark options={errormarkSty}] [each nth point={2},linew, mark=\m, only marks, marksz, color=\c, forget plot] table [x=x\i, y=y\i, y error minus=cl\i, y error plus=ch\i, col sep=comma] {graphics/results/fig_0dB.csv};}
		\temp
	};
	\foreach \i/\c/\m in {11/blueC1/triangle*, 12/blueC1/triangle}{
		\edef\temp{\noexpand \addplot [error bars/.cd, y dir=both, y explicit, error mark options={errormarkSty}] [each nth point={2},linew, mark=\m, only marks, marksz2, color=\c, forget plot] table [x=x\i, y=y\i, y error minus=cl\i, y error plus=ch\i, col sep=comma] {graphics/results/fig_0dB.csv};}
		\temp
	};
	\addlegendimage{short Legend0, color=blueC1};
	\addlegendentry{$\epsilon = 10$\,dB \hspace*{-4pt}};
	\addlegendimage{short Legend0, color=redC1};
	\addlegendentry{$\epsilon = 5$\,dB \hspace*{-4pt}};
	\addlegendimage{short Legend0, color=greenC1};
	\addlegendentry{$\epsilon = 1$\,dB \hspace*{-4pt}};
	\addplot[mark=none, color=black, linew] coordinates {(0,0) (0,1)};
	\label{pgfr1}
	\addplot[mark=none, color=black, dashed, dash pattern=on 5pt off 5pt, dash phase=-0.5pt, linew] coordinates {(0,0) (0,1)};
	\label{pgfr2}
	\addplot[mark=square, only marks, color=black, marksz, linew] coordinates {(0,0) (0,2)};
	\label{pgfr3}
	\addplot[mark=o, only marks, color=black, marksz, linew] coordinates {(0,0) (0,2)};
	\label{pgfr4}
	\addplot[mark=triangle, only marks, color=black, marksz2, linew] coordinates {(0,0) (0,2)};
	\label{pgfr5}
	\addplot[mark=square*, only marks, color=black, marksz, linew] coordinates {(0,0) (0,2)};
	\label{pgfr32}
	\addplot[mark=*, only marks, color=black, marksz, linew] coordinates {(0,0) (0,2)};
	\label{pgfr42}
	\addplot[mark=triangle*, only marks, color=black, marksz2, linew] coordinates {(0,0) (0,2)};
	\label{pgfr52}
	%
%	\conf{{north west}}{0.025}{0.975};
	\end{semilogyaxis}
	\node [draw, fill=white, below=2pt of legendNode.south, anchor=north](n1) {\shortstack[l]{
			\legendfont \ref*{pgfr1}  UE1, analysis \\  
			\,\legendfont \ref*{pgfr32}\,\ref*{pgfr42}\,\ref*{pgfr52}~UE1, simulation \\  
			\legendfont \ref*{pgfr2}  UE2, analysis \\  
			\,\legendfont \ref*{pgfr3}\,\ref*{pgfr4}\,\ref*{pgfr5}~UE2, simulation
	}};
\end{tikzpicture}

%% file: graphics/fig_5dB.tex
\begin{tikzpicture}
	\begin{semilogyaxis}[
	xlabel={Split factor $\alpha$},
	ylabel={Outage probability},
	label style={font=\labelfont},
	ylabel shift = -1mm,
	ylabel right ={~},
%	xlabel shift = -1mm,	
	ymin=0.0001, ymax=1,
	xmin=0, xmax=1,
	ytick pos=left,
	axisSetup,
	ymajorgrids=true,
	xmajorgrids=true,
	yminorgrids=true,
	xminorgrids=true,
	major x grid style={solid, cgrid},
	major y grid style={solid, cgrid},
	minor x grid style={solid, cgrid},
	minor y grid style={gridc},
	legend style={font=\legendfont, name=legendNode, at={(0.195,0.445)},anchor=west},
	legend cell align=left,
%	legend pos=south west,
	]
	% Solid
	\foreach \i/\c in {1/greenC1, 5/redC1, 9/blueC1}{
		\edef\temp{\noexpand \addplot [linew, color=\c, mark=none, forget plot] table [x=x\i, y=y\i, col sep=comma] {graphics/results/fig_5dB.csv};}
		\temp
	};
	% Dashed
	\foreach \i/\c in {2/greenC1, 6/redC1, 10/blueC1}{
		\edef\temp{\noexpand \addplot [linew, color=\c, mark=none, dashed, dash pattern=on 5pt off 5pt, forget plot] table [x=x\i, y=y\i, col sep=comma] {graphics/results/fig_5dB.csv};}
		\temp
	};
	% Error bars
	\foreach \i/\c/\m in {3/greenC1/square*, 4/greenC1/square, 7/redC1/*, 8/redC1/o}{
		\edef\temp{\noexpand \addplot [error bars/.cd, y dir=both, y explicit, error mark options={errormarkSty}] [each nth point={2},linew, mark=\m, only marks, marksz, color=\c, forget plot] table [x=x\i, y=y\i, y error minus=cl\i, y error plus=ch\i, col sep=comma] {graphics/results/fig_5dB.csv};}
		\temp
	};
	\foreach \i/\c/\m in {11/blueC1/triangle*, 12/blueC1/triangle}{
		\edef\temp{\noexpand \addplot [error bars/.cd, y dir=both, y explicit, error mark options={errormarkSty}] [each nth point={2},linew, mark=\m, only marks, marksz2, color=\c, forget plot] table [x=x\i, y=y\i, y error minus=cl\i, y error plus=ch\i, col sep=comma] {graphics/results/fig_5dB.csv};}
		\temp
	};
	\addlegendimage{short Legend0, color=blueC1};
	\addlegendentry{$\epsilon = 10$\,dB \hspace*{-4pt}};
	\addlegendimage{short Legend0, color=redC1};
	\addlegendentry{$\epsilon = 5$\,dB \hspace*{-4pt}};
	\addlegendimage{short Legend0, color=greenC1};
	\addlegendentry{$\epsilon = 1$\,dB \hspace*{-4pt}};
	%
%	\conf{{north west}}{0.025}{0.975};
	\end{semilogyaxis}
	\node [draw, fill=white, below=2pt of legendNode.south, anchor=north](n1) {\shortstack[l]{
			\legendfont \ref*{pgfr1}  UE1, analysis \\  
			\,\legendfont \ref*{pgfr32}\,\ref*{pgfr42}\,\ref*{pgfr52}~UE1, simulation \\  
			\legendfont \ref*{pgfr2}  UE2, analysis \\  
			\,\legendfont \ref*{pgfr3}\,\ref*{pgfr4}\,\ref*{pgfr5}~UE2, simulation
	}};
\end{tikzpicture}

%% file: graphics/fig_10dB.tex
\begin{tikzpicture}
	\begin{semilogyaxis}[
	xlabel={Split factor $\alpha$},
	ylabel={Outage probability},
	label style={font=\labelfont},
	ylabel shift = -1mm,
	ylabel right ={~},
%	xlabel shift = -1mm,	
	ymin=0.0001, ymax=1,
	xmin=0, xmax=1,
	ytick pos=left,
	axisSetup,
	ymajorgrids=true,
	xmajorgrids=true,
	yminorgrids=true,
	xminorgrids=true,
	major x grid style={solid, cgrid},
	major y grid style={solid, cgrid},
	minor x grid style={solid, cgrid},
	minor y grid style={gridc},
	legend style={font=\legendfont, name=legendNode, at={(0.09,0.445)},anchor=west},
	legend cell align=left,
%	legend pos=south west,
	]
	% Solid
	\foreach \i/\c in {1/greenC1, 5/redC1, 9/blueC1}{
		\edef\temp{\noexpand \addplot [linew, color=\c, mark=none, forget plot] table [x=x\i, y=y\i, col sep=comma] {graphics/results/fig_10dB.csv};}
		\temp
	};
	% Dashed
	\foreach \i/\c in {2/greenC1, 6/redC1, 10/blueC1}{
		\edef\temp{\noexpand \addplot [linew, color=\c, mark=none, dashed, dash pattern=on 5pt off 5pt, forget plot] table [x=x\i, y=y\i, col sep=comma] {graphics/results/fig_10dB.csv};}
		\temp
	};
	% Error bars
	\foreach \i/\c/\m in {3/greenC1/square*, 4/greenC1/square, 7/redC1/*, 8/redC1/o}{
		\edef\temp{\noexpand \addplot [error bars/.cd, y dir=both, y explicit, error mark options={errormarkSty}] [each nth point={2},linew, mark=\m, only marks, marksz, color=\c, forget plot] table [x=x\i, y=y\i, y error minus=cl\i, y error plus=ch\i, col sep=comma] {graphics/results/fig_10dB.csv};}
		\temp
	};
	\foreach \i/\c/\m in {11/blueC1/triangle*, 12/blueC1/triangle}{
		\edef\temp{\noexpand \addplot [error bars/.cd, y dir=both, y explicit, error mark options={errormarkSty}] [each nth point={2},linew, mark=\m, only marks, marksz2, color=\c, forget plot] table [x=x\i, y=y\i, y error minus=cl\i, y error plus=ch\i, col sep=comma] {graphics/results/fig_10dB.csv};}
		\temp
	};
	\addlegendimage{short Legend0, color=blueC1};
	\addlegendentry{$\epsilon = 10$\,dB \hspace*{-4pt}};
	\addlegendimage{short Legend0, color=redC1};
	\addlegendentry{$\epsilon = 5$\,dB \hspace*{-4pt}};
	\addlegendimage{short Legend0, color=greenC1};
	\addlegendentry{$\epsilon = 1$\,dB \hspace*{-4pt}};
	%
%	\conf{{north west}}{0.025}{0.975};
	\end{semilogyaxis}
	\node [draw, fill=white, below=2pt of legendNode.south, anchor=north](n1) {\shortstack[l]{
			\legendfont \ref*{pgfr1}  UE1, analysis \\  
			\,\legendfont \ref*{pgfr32}\,\ref*{pgfr42}\,\ref*{pgfr52}~UE1, simulation \\  
			\legendfont \ref*{pgfr2}  UE2, analysis \\  
			\,\legendfont \ref*{pgfr3}\,\ref*{pgfr4}\,\ref*{pgfr5}~UE2, simulation
	}};
\end{tikzpicture}

%% file: graphics/fig_robust.tex
\begin{tikzpicture}
	\begin{axis}[
	xlabel={Outage threshold [dB]},
	ylabel={$\alpha_{\textrm{robust}}$},
	label style={font=\labelfont},
	ylabel shift = -1mm,
	ylabel right ={~},
%	xlabel shift = -1mm,	
	ymin=0, ymax=1,
	xmin=-10, xmax=15,
	ytick pos=left,
	axisSetup,
	ymajorgrids=true,
	xmajorgrids=true,
	yminorgrids=true,
	xminorgrids=true,
	major x grid style={solid, cgrid},
	major y grid style={solid, cgrid},
	minor x grid style={solid, cgrid},
	minor y grid style={gridc},
	legend style={font=\legendfontt, name=legendNode, at={(0.355,0.39)},anchor=west},
	legend cell align=left,
%	legend pos=south west,
	]
	% Solid
	\foreach \i/\c in {1/c1, 2/c2, 3/c3, 4/c4, 5/c5, 6/c6, 7/c7}{
		\edef\temp{\noexpand \addplot [linew, color=\c, mark=none, forget plot] table [x=x\i, y=y\i, col sep=comma] {graphics/results/fig_robust.csv};}
		\temp
	};
	% Dashed5
	\foreach \i/\c/\m in {8/c1/o, 9/c2/square, 10/c3/x, 11/c4/diamond, 12/c5/asterisk, 13/c6/triangle, 14/c7/pentagon}{
		\edef\temp{\noexpand \addplot [linew, color=\c, dashed, mark=o, marksz3, dash pattern=on 5pt off 5pt, forget plot] table [x=x\i, y=y\i, col sep=comma] {graphics/results/fig_robust.csv};}
		\temp
	};
%	\addlegendimage{short Legend3, color=blueC1};
%	\addlegendentry{$\epsilon = 10$\,dB \hspace*{-4pt}};
%	\addlegendimage{short Legend2, color=redC1};
%	\addlegendentry{$\epsilon = 5$\,dB \hspace*{-4pt}};
%	\addlegendimage{short Legend1, color=greenC1};
%	\addlegendentry{$\epsilon = 1$\,dB \hspace*{-4pt}};
	%
%	\conf{{north west}}{0.025}{0.975};
	\addplot[mark=o, color=black, dashed, dash pattern=on 4pt off 3pt, dash phase=-0pt, linew, marksz] coordinates {(0,-0.00001) (0,-0.00001)};
	\label{pgfr22}
	\end{axis}
	\draw (rel axis cs: 0.275, 0.7) -- (rel axis cs: 0.375, 0.65) [->, >=stealth, line width=0.7pt, color=dark011] node[pos=-0.4, above=-8pt] {{\fontsize{8pt}{9pt}\selectfont \contourlength{1.5pt} \contour*{white}{0\,dB}}};
	\draw (rel axis cs: 0.45, 0.87) -- (rel axis cs: 0.47, 0.68) [->, >=stealth, line width=0.7pt, color=dark011] node[pos=-0.2, above=-8pt] {{\fontsize{8pt}{9pt}\selectfont \contourlength{1.5pt} \contour*{white}{5\,dB}}};
	\draw (rel axis cs: 0.84, 0.5) -- (rel axis cs: 0.73, 0.61) [->, >=stealth, line width=0.7pt, color=dark011] node[pos=-0.4, below=-8pt] {{\fontsize{8pt}{9pt}\selectfont \contourlength{1.5pt} \contour*{white}{10\,dB}}};
	\draw (rel axis cs: 0.8, 0.36) -- (rel axis cs: 0.73, 0.49) [->, >=stealth, line width=0.7pt, color=dark011] node[pos=-0.4, below=-8pt] {{\fontsize{8pt}{9pt}\selectfont \contourlength{1.5pt} \contour*{white}{15\,dB}}};
	\draw (rel axis cs: 0.72, 0.225) -- (rel axis cs: 0.665, 0.34) [->, >=stealth, line width=0.7pt, color=dark011] node[pos=-0.4, below=-8pt] {{\fontsize{8pt}{9pt}\selectfont \contourlength{1.5pt} \contour*{white}{20\,dB}}};
	\draw (rel axis cs: 0.625, 0.08) -- (rel axis cs: 0.59, 0.2) [->, >=stealth, line width=0.7pt, color=dark011] node[pos=-0.4, below=-8pt] {{\fontsize{8pt}{9pt}\selectfont \contourlength{1.5pt} \contour*{white}{25\,dB}}};
	\draw (rel axis cs: 0.24, 0.11) -- (rel axis cs: 0.29, 0.21) [->, >=stealth, line width=0.7pt, color=dark011] node[pos=-0.6, below=-8pt] {{\fontsize{8pt}{9pt}\selectfont \contourlength{1.5pt} \contour*{white}{30\,dB}}};
	\node [draw, fill=white, anchor=north west] at (rel axis cs: 0.02,0.98) {\shortstack[l]{
			\legendfont \ref*{pgfr1}  Analysis \\    
			\legendfont \ref*{pgfr22}   Simulation
	}};
\end{tikzpicture}

%% file: Acronyms.tex
\begin{acronym}[DSTTDSGRC]
\setlength{\itemsep}{-3pt}
\acro{CS}{compressed sensing}
\acro{ETF}{equiangular tight frame}
\acro{OGF}{orthoplectic Grassmannian frame}
\acro{NOMA}{non-orthogonal multiple access}
\acro{OMA}{orthogonal multiple access}
\acro{DFT}{discrete Fourier transform}
\acro{CDMA}{code-division multiple-access}
\acro{BCASC}{best complex antipodal spherical codes}
\acro{CBGC}{coherence-based Grassmannian codebook}
\acro{ICBP}{iterative collision-based packing}
\acro{MMSE}{minimum mean squared error}
\acro{MUSA}{multi-user shared access}
\acro{SIC}{successive interference cancellation}
\acro{SNR}{signal-to-noise ratio}
\acro{TDL-C}{tapped-delay-line-C}
\acro{LTE}{long-term evolution}
\acro{SINR}{signal-to-interference-plus-noise ratio}
\acro{SVD}{singular value decomposition}
\acro{KKT}{Karush-Kuhn-Tucker}
\acro{BLER}{block error ratio}
\acro{5G}{fifth-generation}
\acro{6G}{sixth-generation}
\acro{B5G}{beyond fifth-generation}
\acro{IoT}{internet-of-things}
\acro{PAPR}{peak-to-average-power ratio}
\acro{FFT}{fast-Fourier-transform}
\acro{IFFT}{inverse fast-Fourier-transform}
\acro{OFDM}{orthogonal frequency-division multiplexing}
\acro{BS}{base station}
\acro{UE}{user equipment}
\acro{MUD}{multiuser detection}
\acro{CWL}{codeword level}
\acro{MMSE}{minimum mean square error}
\acro{MF}{matched filter}
\acro{PIC}{parallel interference cancellation}
\acro{CRC}{cyclic-redundancy-check}
\acro{RB}{resource-block}
\acrodefplural{RB}{resource-blocks}
\acro{RMS}{root-mean-square}
\acro{DS}{delay spread}
\acro{LDPC}{low-density parity-check}
\acro{MIMO}{multiple-input multiple-output}
\acro{ITS}{intelligent transport systems}
\acro{V2X}{vehicle-to-everything}
\acro{V2V}{vehicle-to-vehicle}
\acro{V2I}{vehicle-to-infrastructure}
\acro{V2N}{vehicle-to-network}
\acro{V2P}{vehicle-to-pedestrian}
\acro{DSRC}{dedicated short-range communication}
\acro{C-V2X}{cellular-\ac{V2X}}
\acro{IEEE}{institute of electrical and electronics engineers}
\acro{MAC}{medium access control}
\acro{PHY}{physical}
\acro{CSMA}{carrier sense multiple access}
\acro{SB-SPS}{sensing-based semi-persistent scheduling}
\acro{5G-NR}{5th generation new-radio}
\acro{mMTC}{massive machine-type communication}
\acro{IGMA}{interleave-grid multiple access}
\acro{IDMA}{interleave-division multiple access}
\acro{ECDF}{empirical cumulative distribution function}
\acro{LLR}{log-likelhood-ratio}
\acro{IRS}{intelligent reflecting surface}
\acro{RIS}{reconfigurable intelligent surface}
\acro{LOS}{line-of-sight}
\acro{NLOS}{non-line-of-sight}
\acro{RV}{random variable}
\acro{CLT}{central limit theorem}
\acro{CDF}{cumulative distribution function}
\acro{IC}{interference cancellation}
\end{acronym}